\documentclass[12pt]{article}

\usepackage{graphicx,float,latexsym,times}
\usepackage{amsfonts,amstext,amsmath,amssymb,amsthm}
\usepackage{amsmath ,amsthm ,amssymb, mathtools, listings}
 \usepackage{makeidx,epsfig,lscape,multirow, titlesec}
 \usepackage{fancybox,color,colordvi,multicol, listings}
 \usepackage{epic,eepic,rotating,picinpar,caption2}
 \usepackage{graphicx,latexsym,xspace,booktabs}
 \usepackage{rccol, hyperref, subfigure, esdiff, setspace}
\usepackage[authoryear]{natbib}
\def\({\left(}
\def\){\right)}
\newcommand{\bi}{\begin{itemize}}
\newcommand{\ei}{\end{itemize}}
\def\ba#1\ea{\begin{align}#1\end{align}}

\def\l{\left\{ }
\def\r{\right\} }
\def\L{\left[ }
\def\R{\right] }
\providecommand{\abs}[1]{\lvert#1\rvert}
\newcommand{\vs}[1]{\vspace{#1 cm}}

\def\Del{\Delta}
\def\sig{\sigma}

\titleformat{\section}{\large\bfseries}{\thesection.}{.5em}{}
\titlespacing*{\section}{0pt}{*3}{*2}
\titleformat{\subsection}{\normalfont\bfseries}{\thesubsection.}{.5em}{}
\titlespacing*{\subsection} {0pt}{*3}{*2}
\titleformat{\subsubsection}{\normalfont\bfseries}{\thesubsubsection.}{.5em}{}
\titlespacing*{\subsubsection} {0pt}{*3}{*2}

\theoremstyle{plain} 
\newtheorem{theorem}{Theorem}[section]
\newtheorem{lemma}{Lemma}[section]

\theoremstyle{definition} 

\numberwithin{equation}{section} 

\begin{document}

\title{\textbf{\Large Minimum Risk Point Estimation of Gini Index}}

\date{}

\maketitle

\vskip -1.4cm
\author{

\vskip -1cm
\begin{flushleft}
 \textbf{ \small SHYAMAL KRISHNA DE} \\
\end{flushleft}
\vs{-.3} 
\par \small School of Mathematical Sciences, National Institute of Science Education and Research, Bhubaneswar, Odisha, India\\ 
\indent (sde@niser.ac.in) 
\begin{flushleft}
\textbf{ \small BHARGAB CHATTOPADHYAY} \footnote{Corresponding author: Department of Mathematical Sciences, FO 2.402A, The University of Texas at Dallas, 800 West Campbell Road, Richardson, TX 75080, USA, E-mail: bhargab@utdallas.edu.}\\
\end{flushleft}
\vs{-.3}
\par \small Department of Mathematical Sciences, The University of Texas at Dallas, Richardson, Texas, USA \\ 
\indent (bhargab@utdallas.edu) \\

}


{\small \noindent\textbf{Abstract:} This paper develops a theory and methodology for estimation of Gini index such that both cost of sampling and estimation error are minimum. Methods in which sample size is fixed in advance, cannot minimize estimation error and sampling cost at the same time. In this article, a purely sequential procedure is proposed which provides an estimate of the sample size required to achieve a sufficiently smaller estimation error and lower sampling cost. Characteristics of the purely sequential procedure are examined and asymptotic optimality properties are proved without assuming any specific distribution of the data. Performance of our method is examined through extensive simulation study.}

\newpage

\section{INTRODUCTION } \label{s:Intro} 
Economic inequality exists in all societies or regions because of the existence of gap in income and wealth among individuals. In order to reduce the gap between the income levels of individuals, government of each and every country devise several economic policies. Periodic evaluation of the effect of economic policies in reducing the income gap between rich and poor is important. There are several inequality indexes in the economic literature. Allison (\citeyear{allison1978measures}) mentioned that among those indices, Gini inequality index is the most widely used measure because it satisfies four basic desirable criteria viz. (i) anonymity, (ii) scale independence, (iii) population independence, and (iv) Pigou-Dalton transfer principle and also Gini index has an easy interpretation and a relation to Lorenz curve.

The most celebrated Gini index, as given in Xu (\citeyear{xu2007}), is
\ba
G_{F}(X)=\frac{\Delta }{2\mu }, \,\, \text{ where }\,\, \Delta =E\left\vert X_{1}-X_{2}\right\vert, \, \mu =E(X)
\label{def:Gini}
\ea
and $X_{1}$ \& $X_{2}$ are two i.i.d. copies of non-negative random variable $X$. If there are $n$ randomly
selected individuals with incomes given by $X_{1},X_{2},\ldots,X_{n}$, then an estimator of \eqref{def:Gini} is given by
\ba
G_{n}=\frac{\widehat{\Delta }_{n}}{2\overline{X}_{n}},
\label{def:GiniEst}
\ea
where $\overline{X}_{n}$ is the sample mean and $\widehat{\Delta }_{n}$ is the Gini's mean difference (GMD) defined as
\ba
\widehat{\Delta }_{n}=\binom{n}{2}^{-1}\sum\limits_{1\leq i_{1}<i_{2}\leq
n}\left\vert X_{i_{1}}-X_{i_{2}}\right\vert .
\label{def:GMD}
\ea
For continuous evaluation of economic policies implemented by the government, periodic computation of Gini index for the whole country or a region is very important. One source from which Gini index of a region or a country can be calculated is using census data which is typically collected every 10 years. But for estimating the Gini index in intermediate years, data from annual household survey conducted by government agencies can be used. For instance, National Sample Survey (NSS) in India, European Statistics on Income and Living Conditions in European Union and other agencies conduct household surveys annually or biennially in respective regions or countries. However, many countries, for example Burundi, Chad, Mozambique (as per world bank website), can not afford or do not collect data from households on a relatively large scale atmost biennially. 

If household survey data is not available, one has to draw a relatively small sample to estimate the Gini index for that region using appropriate sampling technique. The sampling technique should be chosen depending on the size and socio-economic diversity of the country. For a brief review of several sampling techniques, we refer to Cochran (\citeyear{cochran1977g}). In order to compute Gini index for regions or smaller countries, with lesser social diversity, simple random sampling technique can be used to collect income or expenditure data. There exists literature on statistical inference for inequality indices which is computed from household income or expenditure by means of simple random sampling from the population of interest (e.g., Gastwirth \citeyear{gastwirth1972estimation}, Beach and Davidson, \citeyear{beach1983distribution}, Davidson and Duclos, \citeyear{davidson2000statistical}, Xu, \citeyear{xu2007} and Davidson, \citeyear{davidson2009reliable}). In this paper, we will use simple random sampling technique to collect income or expenditure data in order to estimate Gini index accurately.

It is well known that error in estimation decreases or in other words accuracy increases, when the sample size increases. This in turn increases the overall cost of sampling. To minimize the cost of sampling, one has to reduce the sample size which in turn may lead to higher estimation error. Thus, a method of estimation should be developed such that both the cost of sampling and the error in estimation are kept as low as possible. In other words, a procedure is required which can act as a trade-off between the estimation error and the sampling cost. To achieve this trade-off, fixed-sample methodologies can not be used, i.e., the sample size should not be fixed in advance. This problem falls in the domain of sequential analysis where it is known as minimum risk point estimation problem. For more details on the literature of sequential analysis, we refer to Ghosh and Sen ({\citeyear{ghosh1991handbook}), Ghosh et al. (\citeyear{ghosh1997sequential}), Mukhopadhyay and de Silva (\citeyear{mukhopadhyay2009sequential}), and others.

Unlike fixed-sample procedures, sequential procedures do not require sample size to be fixed in advance. Instead, in a sequential procedure, statistical analysis is continued as the observations are collected. Sampling is terminated according to a pre-defined criterion, also known as stopping rule. Sequential sampling allows the estimation process to finish early requiring small sample size. We are certainly not the first one to suggest sequential methods in econometrics. In fact, there are several articles published in several journals in economics and econometrics which pursued the idea of using sequential or multi-stage inference procedures. Examples include \cite{aguirregabiria2007sequential}, \cite{arcidiacono2003finite}, \cite{greene1998gender}, \cite{kanninen1993design}, etc.

Below, we provide a brief literature review of some relevant concepts and also our contribution to the literature of statistical inference and economics.

\subsection{Literature Review and Our Contributions}

The estimator of Gini index in \eqref{def:GiniEst} involves sample mean and Gini's mean difference which belong to a class of unbiased estimators known as U-statistics. Below, we briefly discuss the literature on U-statistics.
\subsubsection{Literature on U-statistics}

The theory and practice of U-statistics began with the pioneering papers of Hoeffding (\citeyear{hoeffding1948class}, \citeyear{hoeffding1961strong}). In the above papers, Hoeffding derived a general method for obtaining unbiased estimators for a parameter $\theta$ associated with an unknown distribution function $F$. Suppose that $X_{1},\ldots,X_{n}$ are \textit{independent and identically distributed} (i.i.d.) random variables from a population with a common distribution function $F$ with an associated parameter $\theta \equiv \theta (F)$, $\theta \in \Theta \subseteq
\mathcal{R}$. Then the U-statistic associated with $\theta $ is
written as follows
\[U\equiv U_{n}^{(m)}=\binom{n}{m}^{-1}\underset{(n, m)}{\sum }
g^{(m)}(X_{i_{1}},...,X_{i_{m}}),\]
where $\underset{(n, m)}{\sum }$ denotes the summation over all possible
combinations of indices $(i_1,\ldots, i_m)$ such that $1\leq i_{1}<i_{2}<\cdots<i_{m}\leq n$, and $m<n$. Here, $
g^{(m)}(.)$ is a symmetric kernel of degree $m$ such that $E_{F}\left[ g^{(m)}(X_{1},...,X_{m})\right] =\theta (F)$ for all $F$. Thus both GMD and the sample mean are U-Statistics with kernels of degree 2 and 1 respectively. Detailed literature on U-statistics can be found in standard textbooks such as Hollander and Wolfe (\citeyear{hollandernonparametric}), Lee (\citeyear{lee1990u}), and others.

Apart from being unbiased estimators, U-statistics are reverse martingales with respect to some non-increasing filtration as proven in Lee (p. 119, \citeyear{lee1990u}). We exploit the reverse martingale property of U-statistics to derive the asymptotic results in section 3. For more literature on reverse martingales, we refer to classical textbooks on probability theory and stochastic processes such as Loeve (\citeyear{loeveprobability}), Doob (\citeyear{doob1953stochastic}), and others.

As discussed before, we estimate the Gini index by a sequential method known as minimum risk point estimation (MRPE). This estimation technique is not new in the literature of sequential analysis. Below, we briefly discuss the developments on minimum risk point estimation.

\subsubsection{Literature on MRPE}

Minimum risk point estimation was first introduced by Robbins (\citeyear{robbins1959}). He suggested a purely sequential procedure for estimating mean of a normal distribution. Ghosh and Mukhopadhyay (\citeyear{ghosh1979sequential}) generalized this idea to a distribution free scenario and developed a purely sequential procedure for minimum risk point estimation of a population mean. Later, Sen and Ghosh (\citeyear{sen1981sequential}) extended the sequential procedure of Ghosh and Mukhopadhyay (\citeyear{ghosh1979sequential}) to accommodate the minimum risk point estimation of any estimable parameter using U-statistics. For more details on MRPE, we refer our readers to Sen (\citeyear{sen1981sequentialbook}), Ghosh et al. (\citeyear{ghosh1997sequential}), Mukhopadhyay and de Silva (\citeyear{mukhopadhyay2009sequential}), and others.

In minimum risk point estimation problems, a cost function is defined which depends on sample size and error in estimation. In this paper, we will use mean square error (MSE) of Gini index as an error in estimation. We are interested in finding an estimate of unknown optimal sample size which minimizes the asymptotic cost function to estimate Gini index of the population.
\subsubsection{Contributions of this paper}

Several fixed-sample methods are developed for estimation of Gini index assuming that the incomes from the sampled individuals are independent and identically distributed (i.i.d.). Examples of such methods can be found in Gastwirth \citeyear{gastwirth1972estimation}, Beach and Davidson, \citeyear{beach1983distribution}, Davidson and Duclos, \citeyear{davidson2000statistical}, Xu, \citeyear{xu2007} and Davidson, \citeyear{davidson2009reliable}. However, these methods cannot be used for minimum risk point estimation of an inequality index. For a brief overview we refer to \cite{Chde2014wseas}. In this article, we propose a sequential procedure that yields an asymptotic minimum risk point estimator of Gini index by minimizing the asymptotic risk function defined as a cost function plus a risk term for estimation error. Under very mild assumptions, we prove that the estimated final sample size for our procedure approaches the theoretically optimal sample size that minimizes the cost function. Moreover, we prove that the expected cost for estimating the Gini index using the estimated final sample size is asymptotically close to theoretically expected cost for estimating the Gini index, that is with theoretically optimal sample size. All theoretical results are validated by extensive simulation study.

The remainder of this paper is organized as follows. Section 2 develops a purely sequential procedure which minimizes both the estimation error and the overall sampling cost. Section 3 presents the theoretical properties enjoyed by the proposed sequential procedure. Performance of our method is assessed via simulation study in Section 4. The next section explores the possibility of satisfying stronger asymptotic optimality properties. In Section 6, we provide some concluding remarks. The appendix contains some auxiliary lemmas and detailed proofs of all theoretical results.

\section{SEQUENTIAL METHOD OF ESTIMATION} \label{s:preliminary}

Suppose incomes from $n$ randomly selected individuals are collected. Let the incomes of $n$ persons be $X_1, \ldots, X_n$ with a common but unknown distribution function $F$. The estimator $G_n$ is a biased estimator of the population Gini index, $G_F$, and $E(G_n -G_F)^2$ is the mean square error (MSE) of $G_n$. The asymptotic expression for MSE of $G_n$ is given by,
\begin{lemma}
\label{lem:asym-risk}
$E(G_F -G_n)^2 = \frac{\xi^2}{n} + O\(\frac{1}{n^{3/2}}\)$, where
\ba
\xi^2 = \frac{\sig_1^2}{\mu^2} + \frac{\Del^2 \sig^2}{4\mu^4} - \frac{\Del}{\mu^3}( \tau - \mu\Del),
\label{def:xi}
\ea
\[
\sig_1^2 = V\L E\( \abs{X_1 - X_2}\,\, \big\vert\, X_1  \) \R, \,\, \tau= E\( X_1 \abs{X_1 - X_2}\), \, \text{ and }\, \sig^2 =V(X),
\]
provided $E(X_1^{12})$ and $E(X_1^{-20})$ exist.
\end{lemma}
The proof of the lemma is given in the appendix. If the sample size is large, we receive more and more information about $G_F$ and, therefore, expect the squared error loss $(G_F -G_n)^2$ due to estimation to be small. However, higher sample size leads to higher sampling cost. Therefore, it is desirable to consider a loss function that takes into account both loss due to error in estimation and the sampling cost. Suppose $c$ is the known cost of sampling each observation. Our goal is to find an estimation procedure which minimizes both the MSE and also the sampling cost. We define a cost function depending on the MSE and the cost of sampling, also known as the risk function, as
\ba
R_n(G_F) = A E(G_F -G_n)^2 + cn.
\label{def:risk}
\ea
Here, $A$ is a known positive constant and is expressed in monetary terms which represents the weight assigned by the researchers or analysts regarding the probable cost per unit squared error loss due to estimation. Thus, the first term $AE(G_F -G_n)^2$ represents the loss in estimating $G_F$ by $G_n$, and the second term $cn$ represents the cost of sampling $n$ observations. The risk function thus gives the expected cost of estimating $G_F$ using the estimator $G_n$ based on incomes from $n$ individuals.  Using the asymptotic expression of MSE of $G_n$ expressed in \eqref{def:xi}, the fixed-sample size risk defined in \eqref{def:risk} becomes
\ba
R_n(G_F) =  A \frac{\xi^2}{n} + cn+ O\(\frac{1}{n^{3/2}}\).
\label{def:asym-risk}
\ea
Thus, \eqref{def:asym-risk} gives the expected cost or the risk, to estimate the unknown value of the population Gini index using $G_n$ based on $n$ observations. Our goal is to find the sample size for which the approximate expected cost (ignoring the $O\(\frac{1}{n^{3/2}}\)$ term) defined in \eqref{def:asym-risk}, i.e.,  $h(n) = A \frac{\xi^2}{n} + cn$ is minimized for all distributions that satisfy the conditions of lemma \ref{lem:asym-risk}.

Considering $n$ as a non-negative continuous variable, the strictly convex function $h(n)$ can be minimized at $n=n_c\left(=\sqrt{\frac{A}{c}}\,\xi\right)$. Thus $n_c$ is the required optimal sample size that should be collected using simple random sampling from the population in order to minimize the expected cost to estimate $G_F$. Thus the approximate expected cost of estimating the Gini index using a sample of size $n_c$ or the asymptotic minimum risk is
\begin{align}
R_{n_c}^{\ast}(G_F)=  A \frac{\xi^2}{n_c} + cn_c = 2cn_c.
\label{def:min-risk}
\end{align}
If the parameter $\xi$ were known in advance, one could simply collect a sample of size $n_c$ which is the minimum sample size to attain the asymptotic minimum risk. Since $\xi$ is not known, we need to collect samples in at least two stages where the first stage is to estimate $\xi$ and $n_c$ based on a pilot sample. In fact, Dantzig (\citeyear{dant1940}) proved that fixed-sample procedures cannot minimize the risk in \eqref{def:asym-risk}, not even asymptotically. Therefore, we propose a purely sequential procedure that yields minimum risk at least asymptotically.

Since $\xi$ is unknown, we first provide an estimator of $\xi$ that is strongly consistent. The estimator of $\xi$ is based on U-statistics
and can also be found in Xu (\citeyear{xu2007}), and Sproule (\citeyear{sproule1969sequential}). Proceeding along the lines of Sproule (\citeyear{sproule1969sequential}), let us define a U-statistic, for each $j=1,2,\ldots,n$,
\begin{equation*}
\widehat{\Delta }_{n}^{(j)}=\binom{n-1}{2}^{-1}\sum\limits_{T_{j}}%
\left\vert X_{i_1} - X_{i_2} \right\vert,
\end{equation*}
where ${T}_{j}{ =\{(i}_{1}{ ,i}_{2}{ ):1\leq i}_{1}%
{<i}_{2}{\leq n}$ and ${i}_{1}{,i}_{2}{ %
\neq j\}}$. Also, define $ W_{jn}{ =n}\widehat{{\Delta }}_{n}{-(n-2)}
\widehat{{\Delta }}_{n}^{(j)}$, for $j=1,\ldots,n$, and $\overline{W}_n = n^{-1}\sum_{j=1}^n W_{jn}$. According to Sproule (\citeyear{sproule1969sequential}), a strongly consistent estimator of $4\sigma_1^2$ is
\[
{s}_{{wn}}^{{ 2}}={(n-1)}^{-1}\sum \limits_{i=1}^{n}( W_{jn} -\overline{W}_n )^2.
\]
Using Xu (2007),
\begin{equation*}
\widehat{\tau }_{n}=\frac{2}{n(n-1)}\sum\limits_{(n,2)}\frac{1}{2}(%
{ X}_{i_{1}}{ +X}_{i_{2}}{ )}\left\vert { X}_{i_{1}}%
{ -X}_{i_{2}}\right\vert
\end{equation*}%
is an estimator of ${ \tau }$. Let $S_n^2$ be the sample variance. Thus, the estimator of ${ \xi }^{%
{ 2}}$\ is%
\begin{equation}
{ V}_{{ n}}^{{ 2}}=\frac{\widehat{\Delta }%
_{n}^{2}S_{n}^{2}}{4\overline{X}_{n}^{4}}-\frac{\widehat{\Delta }_{n}}{%
\overline{X}_{n}^{3}}\widehat{\tau }_{n}+\frac{\widehat{\Delta }_{n}^{2}}{%
\overline{X}_{n}^{2}}+\frac{s_{wn}^{2}}{4\overline{X}_{n}^{2}}.
\label{est-of-xi2}
\end{equation}
Using Sproule (1969) and theorem 3.2.1 of Sen (1981, p. 50), we conclude that $V_{n}^{2}$ is a strongly consistent
estimator of $\xi ^{2}$.

We outline the purely sequential estimation procedure of the Gini Index of the population as follows:\\
Step 1:  In the first step, often called the pilot sample step, incomes from a sample of m individuals are collected.
This sample is called the pilot sample. Based on this pilot sample of size $m$, an estimate of $\xi^2$ obtained by computing ${V}_{{m}}^{{2}}$. Check the condition, $m\geq \sqrt{\frac{A}{c}}V_{m}$. If $m < \sqrt{\frac{A}{c}}V_{m}$ then go to the next step. Otherwise, if $m \geq \sqrt{\frac{A}{c}}V_{m}$, then stop sampling and set the the final sample size equal to $m$.
\\
Step 2: Obtain income from one randomly selected individuals. Update the estimate of $\xi^2$ and verify the condition based on $m+1$ observations.  If $m+1 \geq \sqrt{\frac{A}{c}}V_{m+1}$ stop further sampling and set the final sample size equal to $m+1$. If $m+1 < \sqrt{\frac{A}{c}}\left(V_{m+1}\right)$ then continue the sampling process by sampling $1$ more individuals and simultaneously update the condition.
\\
The sampling process is continued until the updated condition is satisfied.\\
Formally, we define the stopping rule $N$, for every $c>0$, as
\begin{equation}
\label{stopping-rule1}
N \equiv N(c) \, \text{ is the smallest integer }\, n (\geq m) \text{ such that } %
n\geq \sqrt{\frac{A}{c}}V_{n}.
\end{equation}%
Here, $m$ is the initial or pilot sample size. In some extreme situations, the estimator $V_n$ may be very small which may cause our procedure to stop too early. To avoid this problem, we propose a slightly modified stopping rule $N_{c}$ as
\begin{equation}
N_{c}\, \text{ is the smallest integer }\, n(\geq m) \ni  n\geq \sqrt{\frac{A}{c}}%
\left( V_{n}+n^{-\gamma }\right),
\label{stopping-rule}
\end{equation}
where $\gamma\in(0,0.5)$ is a suitable constant. The inclusion of the term $n^{-\gamma }$ ensures that we do not stop too early due to small value of $V_n$.

\section{THEORETICAL RESULTS}\label{results}

For a given cost $c$ per observation, the risk or the expected cost for estimating the Gini index $G_F$ using an estimator based on the final sample size $N_{c}$ is given by
\begin{equation}
R_{N_{c}}(G_{F})=AE(G_{F}-G_{N_{c}})^{2}+cE(N_{c}).
\label{risk-at-stopping}
\end{equation}
Thus, the estimator $G_{N_{c}}$ is asymptotically minimum risk point estimator (AMRPE) if the ratio regret is
asymptotically 1, i.e., if%
\begin{equation}
\underset{c\rightarrow 0}{\lim }R_{N_{c}}(G_{F})/R_{n_{c}}(G_{F})=1.
\label{ratio-regret}
\end{equation}
In other words, estimator $G_{N_{c}}$ is AMRPE (refer Sen, 1981) if the expected cost for estimating the Gini index $G_F$ using an estimator based on the final sample size $N_{c}$ is asymptotically close to expected cost for estimating $G_F$ using the optimal sample size, $n_c$. In decision theoretic framework, the ratio in \eqref{ratio-regret} is known as ratio regret which is the ratio between the actual payoff and the minimum payoff due to some optimal strategy (Loomes and Sugden, \citeyear{loomes1982regret}). 

Before discussing the asymptotic optimality properties of our method, we prove in the following lemma that if observations are collected using \eqref{stopping-rule}, sampling will stop at some finite time with probability one.
\begin{lemma}
\label{lem:sample}
Under the assumption that $\xi <\infty$, for any $c>0$, the stopping time $N_c$ is finite, i.e., $P(N_c < \infty) =1$.
\end{lemma}
Proof of this lemma is given in Appendix. This lemma is very crucial for any sequential procedure because it assures that the practitioner will not need to sample indefinitely. Below we provide the main theorem related to the asymptotic optimality properties of our procedure.
\begin{theorem}
\label{thm:main}
The stopping rule \eqref{stopping-rule} yields:
\begin{itemize}
\item[(i)] $N_{c}/n_{c}\to 1$ almost surely as $c \downarrow 0$.

\item[(ii)] $E(N_{c}/n_{c})\to 1$ as $c\downarrow 0$. [Asymptotic First-order Efficiency]

\item[(iii)] If $\gamma \in (0, \frac{1}{2})$, $R_{N_{c}}(G_{F})/R_{n_{c}}(G_{F})\to 1$ as $c\downarrow 0$. [Asymptotic First-order Risk Efficiency]
\end{itemize}
provided, $E(X^{16})$ and $E(X^{-24})$ exist.
\end{theorem}
\begin{proof}
Proof of this theorem is technical and, therefore, it is given in Appendix.
\end{proof}
The parts (i) and (ii) of this theorem imply that the final sample size of our procedure is asymptotically same as the minimum sample size required to minimize the asymptotic risk defined in \eqref{def:asym-risk}. The part (iii) proves that the risk attained by our procedure is asymptotically same as the minimum risk. Therefore, the Gini index estimator $G_{N_c}$ is indeed AMRPE. The optimality properties in part (ii) and (iii) are well known in the sequential literature as asymptotic first-order efficiency and asymptotic first-order risk efficiency respectively (see Mukhopadhyay and de Silva, 2009). Theorem \ref{thm:main} also holds for the stopping rule defined in \eqref{stopping-rule1}.


\section{PERFORMANCE VIA SIMULATIONS} \label{simulations}

In this section, we evaluate performance of our estimation strategy for moderate sample size (i.e., $c$ is small but not too small) via simulation study.
\begin{table}
\centering
\caption{Estimated sample variances and covariances}
\vs{.4}
\begin{tabular}{|c|c|c|c|c|c|c|}
\hline 
 &&&&&& \\ [-1.2ex]
Distribution & $\,\;\,%
\begin{tabular}{c}
$\overline{s_{{ wN}}^{2}}$ \\ [1ex]
\textbf{s(}$s_{{ wN}}^{2}$\textbf{)}%
\end{tabular}%
$ & $4\sigma _{1}^{2}$ & $%
\begin{tabular}{c}
$\overline{\widehat{{ \tau }}_{N}}$ \\ [1ex]
\textbf{s(}$\widehat{{ \tau }}_{N}$\textbf{)}%
\end{tabular}%
$ & $\tau $ & $%
\begin{tabular}{c}
$\overline{{ V}_{{ N}}^{{ 2}}}$ \\ [1ex]
\textbf{s(}$V_{{ N}}^{{ 2}}$\textbf{)}%
\end{tabular}%
$ & $\xi ^{2}$ 
\\
&&&&&& \\ [-1.2ex] 
\hline
&&&&&& \\ [-1.2ex]
Exponential & $%
\begin{tabular}{c}
0.0521 \\
0.0002%
\end{tabular}%
$ & 0.0532 & $%
\begin{tabular}{c}
0.0596 \\
0.0001%
\end{tabular}%
$ & 0.06000 & $%
\begin{tabular}{c}
0.0843 \\
0.0002%
\end{tabular}%
$ & 0.0833 
\\
&&&&&& \\ [-1.2ex] 
\hline
&&&&&& \\ [-1.2ex]
Gamma & $%
\begin{tabular}{c}
3.4172 \\
0.0157%
\end{tabular}%
$ & 3.5036 & $%
\begin{tabular}{c}
7.8110 \\
0.0147%
\end{tabular}%
$ & 7.8205 & $%
\begin{tabular}{c}
0.0463 \\
0.0001%
\end{tabular}%
$ & 0.0468 
\\
&&&&&& \\ [-1ex] 
\hline
&&&&&& \\ [-1ex]
Lognormal & $%
\begin{tabular}{c}
52.11274 \\
0.1173%
\end{tabular}%
$ & 52.8108 & $%
\begin{tabular}{c}
84.9292 \\
0.0694%
\end{tabular}%
$ & 85.2236 & $%
\begin{tabular}{c}
0.0498 \\
0.00009%
\end{tabular}%
$ & 0.0526 
\\
\hline
\end{tabular}%
\end{table}

\begin{table}
\centering
\caption{ Estimated average final sample size and the ratio regret}
\vs{.4}
\begin{tabular}{|c|c|c|c|c|c|c|}
\hline
 &&&&&& \\ [-1.2ex]
Distribution &
\begin{tabular}{l}
$\,\;\,\overline{N}$ \\  [1ex]
$\underset{}{\mathbf{s(}\overline{N}\mathbf{)}}$%
\end{tabular}
& $n_{c}$ & $\overline{N}/n_{c}$ & $\max (N)$ &
\begin{tabular}{c}
$\overline{{ r}}_{{ N}}$ \\  [1ex]
s($\overline{{ r}}_{{ N}}$)%
\end{tabular}
& $\frac{\overline{{ r}}_{{ N}}}{{ R^{\ast}}_{n_{c}}}$ 
\\
&&&&&& \\ [-1.4ex] 
\hline
&&&&&& \\ [-1.2ex]
Exponential &
\begin{tabular}{c}
205.4111 \\
0.2378%
\end{tabular}
& 204.08 & 1.0065 & 319 &
\begin{tabular}{c}
40.9317 \\
0.0474%
\end{tabular}
& 1.0028 
\\
&&&&&& \\ [-1.2ex] 
\hline
&&&&&& \\ [-1.2ex]
Gamma &
\begin{tabular}{c}
152.19 \\
0.1970%
\end{tabular}
& 152.97 & 0.9949 & 239 &
\begin{tabular}{c}
30.2765 \\
0.0391%
\end{tabular}
& 0.9904 
\\
&&&&&& \\ [-1.2ex] 
\hline
&&&&&& \\ [-1.2ex]
Lognormal &
\begin{tabular}{c}
162.3504 \\
0.1483%
\end{tabular}
& 163.10 & 0.9954 & 228 &
\begin{tabular}{c}
152.07 \\
0.1958%
\end{tabular}
& 0.9919 \\ \hline
\end{tabular}%
\end{table}
\noindent To implement the sequential procedure in \eqref{stopping-rule1}, we fix $c =0.1$, $A=50000$, and the pilot sample size $m=10$. The results in Table 1 and 2 are based on random samples from three income distributions: exponential (rate $=5$), gamma (shape $=2.649$, rate $=0.84$), and log-normal (mean $=2.185$, sd $=0.562$). Number of replications used in all Monte carlo simulations is 5000. Table 1 compares the true values of the parameters $\sigma_1^2$, $\tau$, and $\xi^2$ with their estimated values based on the final sample size $N$. $s\(s^2_{wN}\)$, $s\( \hat{\tau}_N \)$, and $s\(  V^2_N\)$ represent the standard errors of the estimators $s^2_{wN}$, $\hat{\tau}_N$, and $V^2_N$ respectively.

Table 1 shows that the average values of the estimators are close to
the true values of the parameters and, therefore, it indicates that $s_{wN}^{2}\rightarrow 4\sigma _{1}^{2}$, $\widehat{{ %
\tau }}_{N}\rightarrow \tau $ and ${ V}_{{ N}}^{{ 2}%
}\rightarrow \xi ^{2}$ as $c\downarrow 0.$

Table 2 presents the average final sample size $\overline{N}$ (estimates $E(N)$), the maximum sample size $\max (N)$ from 5000 replications, and the average risk $\overline{r}_N$ (estimates $R_N(G_F)$) obtained from the sample of size $N$. Moreover, $s(\overline{N})$ and $s(\overline{r}_N)$ represent the standard errors of $\overline{N}$ and $\overline{r}_N$ respectively. Table 2 shows that the average sample size $\overline{N}$ is almost the same as the optimal sample size $n_c$. Therefore, on average, our procedure requires only the minimum sample size $n_c$. The last column of Table 2 illustrates that, on average, the cost for estimating the Gini index $G_F$ using an estimator based on the estimated final sample size is asymptotically close to expected cost for estimating $G_F$ using the optimal sample size, $n_c$, or in other words, the ratio regret is very close to 1. This implies that the risk incurred by our method is almost the same as the minimum possible risk $R^{\ast}_{n_c}$ defined in \eqref{def:min-risk}. Thus, we find that the proposed sequential procedure performs remarkably well for the above mentioned income distributions.

\section{EXTENSIONS AND DISCUSSIONS}


\subsection{Exploring Asymptotic Second-Order Efficiency}

In sequential point estimation literature, a stopping rule $N_c$ is known as \emph{asymptotically second-order efficient} (see Ghosh and Mukhopadhyay, \citeyear{ghosh1981consistency}) if the difference between the expected final sample size $E(N_c)$ and the theoretically optimum fixed-sample size $n_c$ is asymptotically bounded, i.e., if $E(N_c) - n_c$ is bounded  as $c \downarrow 0$. Clearly, if a sequential method is second-order efficient, it is first-order efficient as well. However, the converse is not necessarily true. We explore this second-order efficiency property via Monte Carlo simulations. Under the same scenario as in Tables 1 and 2, we apply our method and estimate the difference $E(N_c) - n_c$ based on 500 replications. We repeat this process 10 times and present 10 observed values of $\overline{N}-n_c$ each estimating $E(N_c) - n_c$.
\begin{table}
\centering
\caption{ Estimated Values of $E[N_{c}]-n_{c}$ }
\vs{.4}
\begin{tabular}{|c|ccccc|}
\hline
 &&&&& \\ [-1.2ex]
Distribution & \multicolumn{5}{c|}{$E[N_{c}]-n_{c}$} 
\\
&&&&& \\ [-1.2ex] 
\hline
&&&&& \\ [-1.2ex]
Exponential & 1.9800 & 0.7660 & 1.5520 & 2.600 & 1.0380 \\
& 1.1560 & 0.6760 & 1.5100 & 0.7020 & 1.9020 
\\
&&&&& \\ [-1.2ex] 
\hline
&&&&& \\ [-1.2ex]
Gamma & -1.178 & -0.98 & -0.274 & -1.338 & -1.2 \\
& -0.834 & -0.438 &  -0.686 & -0.282 & -0.626 
\\
&&&&& \\ [-1.2ex] 
\hline
&&&&& \\ [-1.2ex]
Lognormal & -0.7211 & -0.9591 & -0.8331 & -0.7951 & 0.6771 \\
& -0.6611 & -0.1091 & -1.2671 & -1.2211 & -0.2031 \\
\hline
\end{tabular}%
\end{table}
Table 3 shows that the differences $E(N_c) - n_c$ are quite small for all three distributions. Therefore, simulation study strongly indicates that the proposed sequential procedure is asymptotically second-order efficient.


\subsection{Exploring Asymptotic Second-Order Risk Efficiency}

In sequential point estimation literature, a stopping rule $N_c$ is known as \emph{asymptotically second-order risk efficient} (see Ghosh and Mukhopadhyay, \citeyear{ghosh1981consistency}) if the difference regret, i.e.,  $R_{N_{c}}(G_{F})- R_{n_c}(G_F)$ is asymptotically bounded. This property implies asymptotic first-order risk efficiency. We explore this second-order risk efficiency property via Monte Carlo simulations. For each of the three distributions in Table 4, 10 observed values of $\overline{r}_{N_c} - R_{n_c}$ are presented, each estimating $R_{N_{c}}(G_{F})- R_{n_c}(G_F)$. Table 4 shows that the differences $R_{N_{c}}(G_{F})-R_{n_{c}}(G_F)$ are quite small for all three distributions. Monte Carlo simulations strongly indicates that the proposed sequential procedure is asymptotically second-order risk efficient.
\begin{table}
\centering
\caption{ Estimated Values of $R_{N_{c}}(G_{F})-R_{n_{c}}(G_F)$}
\vs{.4}
\begin{tabular}{|c|ccccc|}
\hline
 &&&&& \\ [-1.2ex]
Distribution & \multicolumn{5}{c|}{$\underset{}{\overset{}{\overline{{ %
r}}_{{ N_c}}-{ R}_{n_{c}}}}$ values}
\\
&&&&& \\ [-1.2ex] 
\hline
&&&&& \\ [-1.2ex]
Exponential & 0.2397 & -0.0016 & 0.1568 & 0.3799 & 0.0619 \\
& 0.0796 & -0.0150 & 0.1564 & -0.0165 & 0.2302 
\\
&&&&& \\ [-1.2ex] 
\hline
&&&&& \\ [-1.2ex]
Gamma$\,\;$ & -0.3733 & -0.3387 & -0.1901 & -0.4009 & -0.3769 \\
& -0.3048 & -0.2218 & -0.2768 & -0.1933 & -0.2547 
\\
&&&&& \\ [-1.2ex] 
\hline
&&&&& \\ [-1.2ex]
Lognormal & -0.3080 & -0.2878 & -0.2796 & -0.2480 & -0.2494 \\
& -0.3704 & -0.3587 & -0.1540 & -0.2625 & -0.1379 \\ 
\hline
\end{tabular}
\end{table}


%
%


\section{CONCLUDING REMARKS}

The Gini index or Gini concentration is a very popular measure of inequality. It is well known that error in estimation of Gini index decreases when the sample size increases. This inflates the overall cost of sampling. In order to compute Gini index for a region or a smaller country with lesser diversity at a specific point of time, we develop a procedure which computes the final sample size needed to minimize both the error of estimation as well as the cost of sampling via simple random sampling technique.

Without assuming any specific distribution for the data, we showed that the average final sample size using our procedure approaches the unknown optimal sample size that minimizes the cost function. Moreover, we proved that the expected cost for estimating the Gini index using the estimated final sample size is asymptotically close to the expected cost for estimating the Gini index using the unknown optimal sample size. Thus, based on the results mentioned above, we conclude that the proposed sequential estimation strategy is remarkably efficient in reducing both sampling cost and estimation error.

\section{APPENDIX: AUXILIARY RESULTS AND PROOFS}

\subsection{Proof of Lemma \ref{lem:sample}}
Note that $V_n$ is strongly consistent estimator of $\xi$. Therefore, for any fixed $c>0$,
\begin{align*}
P(N_c >\infty) &= \lim_{n\to \infty} P( N_c >n) \\
&= \lim_{n \to \infty} P\( n < \sqrt{A/c}\,(V_n + n^{-\gamma}) \) = 0.
\end{align*}
The last equality is obtained since $V_n \to \xi$ almost surely as $n \to \infty$. This completes the proof.

\subsection{Lemmas to Prove The Main Result}

This section is dedicated to prove some lemmas that are essential to establish the main theorem \ref{thm:main}. First, we introduce few notations. Note from \eqref{stopping-rule} that $N_c \ge \sqrt{\frac{A}{c}} \, N_c^{-\gamma}$, i.e., $N_c \ge \(\frac{A}{c}\)^{\frac{1}{2(1+\gamma)}}$ with probability 1. For fixed $\epsilon, \, \gamma > 0$, define
\ba
n_{1c}=\(\frac{A}{c}\)^{\frac{1}{2(1+\gamma)}}, \quad n_{2c} = n_c(1-\epsilon), \,\, \text{ and } \,\, n_{3c} = n_c(1+\epsilon), \, \text{ where } \, n_c = \sqrt{\frac{A}{c}}\, \xi.
\label{defs}
\ea 
Suppose $\boldsymbol{X}_{(n)}$ denotes the $n$ dimensional vector of order statistics from the sample $X_1, \ldots, X_n$, and $\mathcal{F}_n$ is the $\sigma $-algebra generated by $( \boldsymbol{X}_{(n)}, X_{n+1}, X_{n+2}, \ldots).$ By Lee (1990), $\left\{ \overline{{ X}}_{n}{
,\mathcal{F} }_{n}\right\} $, $\left\{ { S}_{{ n}}^{{ 2}}%
{ ,\mathcal{F} }_{n}\right\} $, $\left\{ \text{$\widehat{{ \tau
}}_{n}{ ,\mathcal{F} }_{n}$}\right\} $, $\left\{ \widehat{{ %
\Delta }}_{n},{ \mathcal{F} }_{n}\right\} $, and their convex functions are all reverse submartingales. Using reverse submartingale properties of U-statistics, we prove the following maximal inequality for sample Gini's mean difference. 
%
\begin{lemma} \label{A1}
If nonnegative i.i.d. random variables $X_1, \ldots, X_n$ are from the distribution $F$ such that $E(X_1^{\max{(2r,\, p)} }) < \infty$ for some positive integers $r$ and $p$, then for any $k>0$,
\begin{align*}
{ P}\left( \underset{n_{1c}\leq n\leq n_{2c}}{\max }\left\vert
\widehat{\Delta }_{n}^{2}-\Delta ^{2}\right\vert \geq k\right) { \leq
O(n}_{1c}^{-r/2}{ )+O(n}_{1c}^{-p/2}{ )} \,\, \text{ as } c\downarrow 0. 
\end{align*}
\end{lemma}
\begin{proof}
Note that
\begin{align}\label{A12}
\left\vert \widehat{\Delta }_{n}^{2}-\Delta ^{2}\right\vert  &  =
\left\vert \left( \widehat{\Delta }_{n}^{2}-\Delta ^{2}\right) I(\widehat{
\Delta }_{n}>\Delta )+\left( \widehat{\Delta }_{n}^{2}-\Delta ^{2}\right) I(
\widehat{\Delta }_{n}<\Delta )\right\vert  \notag
\\
& { \leq }\left( \widehat{\Delta }_{n}^{2}-\Delta ^{2}\right) ^{+}%
{ + \,2 \Delta }\left\vert \widehat{\Delta }_{n}-\Delta \right\vert
{ I(}\widehat{\Delta }_{n}{ <\Delta )}. 
\end{align}
Here, the notation $x^{+}$ is used to mean $\max(x, 0)$. Therefore,
\begin{align*}
P\left( \underset{n_{1c}\leq n\leq n_{2c}}{\max }\left\vert \widehat{\Delta }_{n}^{2}-\Delta ^{2}\right\vert \geq k\right)  
\leq  
P\left( \underset{n_{1c}\leq n\leq n_{2c}}{\max }\left( \widehat{%
\Delta }_{n}^{2}-\Delta ^{2}\right) ^{+}\geq \frac{k}{2}\right) +
P\left( \underset{n_{1c}\leq n\leq n_{2c}}{\max }\left\vert
\widehat{\Delta }_{n}-\Delta \right\vert \geq \frac{k}{4\Delta }\right).
\end{align*}
Since $\left( \widehat{\Delta }_{n}^{2}{ -\Delta }%
^{2}\right) ^{+}$ and $\left\vert \widehat{\Delta }_{n}{ -\Delta }%
\right\vert $ are reverse submartingales, using maximal
inequality for reverse submartingales (Ghosh et al. 1997), we write 
\begin{align}
P\( \max_{n_{1c} \le n \le n_{2c}} \left\vert \widehat{\Delta }_n^2 - \Delta ^2 \right\vert \ge k \) 
& \le
\(\frac{2}{k}\)^r E\left[\( \( \widehat{\Delta }_{n_{1c}}^2-\Delta ^{2}\) ^{+}
\)^r \right] + \(\frac{4\Delta}{k}\)^p  E \left[ \left\vert \widehat{\Delta }_{n_{1c}}-\Delta \right\vert^p \right] \notag 
\\
 & \le \(\frac{2}{k}\)^r \l E\left[\( \widehat{\Delta }_{n_{1c}}-\Delta \)^{2r}\right] E\left[\( \widehat{\Delta }_{n_{1c}}+\Delta \)^{2r}\right] \r^{\frac{1}{2}} + O\(n_{1c}^{\frac{-p}{2}}\) \notag
\\
& \le O\(n_{1c}^{-r/2}\) + O\(n_{1c}^{-p/2}\). \notag
\end{align}
The last two inequalities are obtained by Cauchy--Schwarz inequality and lemma 2.2 of Sen and Ghosh (1981). The moment conditions of this lemma are needed to ensure that all expectaions exist in the last three inequalities. 
\end{proof}
%
\begin{lemma}\label{A8}
Let $\overline{X}_n$ be the sample mean based on nonnegative i.i.d. observations $X_1, \ldots, X_n$. For $r \ge 1$, $E\( \overline{X}_n^{\,\, -r} \) \le E\( X_1^{\,-r}\)$.
\end{lemma}
\begin{proof}
Note that $\overline{X}_n \ge \( \prod_{i=1}^n X_i \)^{1/n}$ as the observations are nonnegative. Therefore,
\begin{align}\label{A81}
E\( \overline{X}_n^{\,\, -r} \) \le E\left[ \( \prod_{i=1}^n \frac{1}{X_i}\)^{r/n} \right] = \l E\left[ \( \frac{1}{X_1}\)^{r/n}\right]\r^n.
\end{align} 
The last equality is due to the i.i.d. property of the observations. We know that $ \l E\( |X|^p\) \r^{1/p}$ is a nondecreasing function of $p$ for $p>0$. Applying this result with $p=1/n \ge 1$ in \eqref{A81}, we complete the proof.
\end{proof}
\begin{lemma} \label{A2}
 Suppose that nonnegative i.i.d. random variables $X_1, \ldots, X_n$ are observed from the distribution $F$ such that $E(X_1)^{4p}$ and $E(X_1)^{-\max{\l 4p,\, 2p(r-1) \r}}$ exist for some positive integers $r$ and $p$. Then, for any $k>0$,
\begin{align*}
P\left( \underset{n_{1c}\leq n\leq n_{2c}}{\max }\left\vert \frac{1}{%
\overline{X}_{n}^{r}}-\frac{1}{\mu ^{r}}\right\vert \geq k\right) \le O(n_{1c}^{-p/2}) \,\, \text{ as } \, c \downarrow 0.
\end{align*}
\end{lemma}
\begin{proof} By Taylor expansion of $\overline{X}_n^{-r} = \frac{1}{\mu^r}\( 1+ (\overline{X}_n -\mu)/\mu\)^{-r}$, we have
\begin{align*}
\left\vert \( \frac{1}{\overline{X}_n^r}  - \frac{1}{\mu^r}\) I\( \frac{1}{\overline{X}_n} < \frac{1}{\mu}\)  \right\vert 
= \frac{1}{\mu^r} \left\vert \l -\frac{r}{\mu}(\overline{X}_n - \mu) + \frac{r(r+1)}{2\mu^2} \frac{(\overline{X}_n - \mu) ^2}{ z^{r+2}} \r   I\( \frac{1}{\overline{X}_n} < \frac{1}{\mu}\) \right\vert,
\end{align*}
where $z \in [ 1, \overline{X}_n / \mu]$. Since $z^{-(r+2)}I(\overline{X}_n^{-1} < \mu^{-1}) \le 1$, proceeding along the lines of \eqref{A12}
\begin{align}\label{A21}
\left\vert  \frac{1}{\overline{X}_n^r}  - \frac{1}{\mu^r} \right\vert & = \left\vert \( \frac{1}{\overline{X}_n^r}  - \frac{1}{\mu^r}\) I\( \frac{1}{\overline{X}_n} \ge \frac{1}{\mu}\) + \( \frac{1}{\overline{X}_n^r}  - \frac{1}{\mu^r}\) I\( \frac{1}{\overline{X}_n} < \frac{1}{\mu}\)  \right\vert  \notag
\\
& \le  \( \frac{1}{\overline{X}_n^r}  - \frac{1}{\mu^r}\)^{+} +  \frac{r}{\mu^{r+1}} \left\vert \overline{X}_n - \mu \right\vert  + \frac{r(r+1)}{2\mu^{r+2}} (\overline{X}_n - \mu) ^2.  
\end{align}
Let $U_{1n} =\( \frac{1}{\overline{X}_n^r}  - \frac{1}{\mu^r}\)^{+}$, $U_{2n}=  \frac{r}{\mu^{r+1}} \left\vert \overline{X}_n - \mu \right\vert$, and $U_{3n} =  \frac{r(r+1)}{2\mu^{r+2}} (\overline{X}_n - \mu) ^2$.  Using \eqref{A21}, we can write
\begin{align}\label{A22}
P\left( \underset{n_{1c}\leq n\leq n_{2c}}{\max }\left\vert \frac{1}{\overline{X}_{n}^{r}}-\frac{1}{\mu ^{r}}\right\vert \geq k\right)  
 & \le P\left( \underset{n_{1c}\leq n\leq n_{2c}}{\max } U_{1n}\ge \frac{k}{3}\right)  +  P\left( \underset{n_{1c}\leq n\leq n_{2c}}{\max } U_{2n}\ge \frac{k}{3}\right) \notag \\
&  +  P\left( \underset{n_{1c}\leq n\leq n_{2c}}{\max } U_{3n}\ge \frac{k}{3}\right). 
\end{align}
Since $\( \frac{1}{\overline{X}_n^r}  - \frac{1}{\mu^r}\)$ is a reverse submartingale and $f(x)=x^{+}$ is a non-decreasing convex function of $x$, $U_{1n}$ is a reverse submartingale. Therefore, using maximal inequality for reverse submartingales
\begin{align}\label{A23}
 P\left( \max_{n_{1c} \le n \le n_{2c}} U_{1n}\ge \frac{k}{3}\right)
 & \le \left( \frac{3}{k}\right) ^p E\left[ \left( \frac{1}{
\overline{X}_{n_{1c}}^{r}}-\frac{1}{\mu ^{r}}\right) ^{+}\right] ^{p} \notag
\\
 &\le \left( \frac{3}{k}\right) ^p E\left[ \left( \frac{1}{
\overline{X}_{n_{1c}}}-\frac{1}{\mu }\right) \left( \frac{1}{\overline{X}
_{n_{1c}}^{r-1}}+\frac{1}{\mu \overline{X}_{n_{1c}}^{r-2}}+...+\frac{1}{\mu
^{r-1}}\right) I\(\overline{X}_{n_{1c}}  < \mu \)\right] ^p \notag
\\
 & \le \left( \frac{3}{k}\right) ^{p}{ r}^{p} E\left[
\left( \frac{1}{\overline{X}_{n_{1c}}}-\frac{1}{\mu }\right) ^{p}\overline{X}%
_{n_{1c}}^{-p(r-1)}\right] \notag
\\
& \le  \(\frac{3r}{k}\)^p \l E\left[\( \overline{X}_{n_{1c}} - \mu \)^{4p}\right] E \left[ \( \frac{1}{\mu \overline{X}_{n_{1c}} } \)^{ 4p} \right] \r^{\frac{1}{4}} \l E\( \frac{1}{\overline{X}_{n_{1c}}^{2p(r-1)} } \) \r^{\frac{1}{2}}  \notag
\\
 & \le O(n_{1c}^{-p/2}).
\end{align}
The last two inequalities are obtained by using Cauchy--Schwarz inequality and lemma 2.2 of Sen and Ghosh (1981). Due to lemma \ref{A8}, existence of  $E(X_1)^{-\max{\l 4p,\, 2p(r-1) \r}}$ ensures the existence of $ E\left[\( \frac{1}{ \overline{X}_{n_{1c}} } \)^{ 4p}\right]$ and $ E\left[\( \frac{1}{\overline{X}_{n_{1c}} } \)^{ 2p(r-1)}\right]$. Since $\lvert \overline{X}_n - \mu \rvert$ and $(\overline{X}_n - \mu)^2$ are reverse submartingales, we can write
\begin{align}\label{A24}
 P\left( \underset{n_{1c}\leq n\leq n_{2c}}{\max } U_{2n} \geq \frac{k}{3}\right) 
 & \le \left( \frac{3r}{k\mu ^{r+1}}\right) ^{2p} E\left(
\overline{X}_{n_{1c}}-\mu \right) ^{2p} \le O(n_{1c}^{-p}), \\
 P\left( \underset{n_{1c}\leq n\leq n_{2c}}{\max } U_{3n}\ge \frac{k}{3}\right) 
& \le \left( \frac{3r(r+1)}{2k\mu ^{r+2}}\right)^p E\left(
\overline{X}_{n_{1c}}-\mu \right) ^{2p} \le O(n_{1c}^{-p}). \label{A25}
\end{align}
Apply \eqref{A23}, \eqref{A24}, and \eqref{A25} in \eqref{A22} to complete the proof. 
\end{proof}
%
\begin{lemma}\label{A3}
 Suppose that nonnegative i.i.d. observations $X_1, \ldots, X_n$ are such that $E(X_1^{4r})$ and $E(X_1^{-6r})$ exist for some $r\ge 1$. For any $\epsilon \in (0, 1)$ and $\gamma >0$, 
\begin{itemize}
\item[(i)] $P(N_{c}\le n_{c}(1-\epsilon ))= O\(n_{1c}^{ -\frac{r}{2} }\) = O\left( c^{\frac{r}{4(1+\gamma)}}\right)\,$ as $\, c \downarrow 0$,
\item[(ii)] $P(N_{c}\ge n_{c}(1+\epsilon ))= O\(n_{1c}^{ -\frac{r}{2} }\) = O\left( c^{\frac{r}{4(1+\gamma)}}\right)\,$ as $\,c \downarrow 0$.
\end{itemize}
\end{lemma}
\begin{proof}
Using the definition of stopping rule $N_c$ in \eqref{stopping-rule} and \eqref{defs}, we have 
\begin{align}\label{A31}
 P( N_c \le n_{2c} ) & \le P\( n > \sqrt{\frac{A}{c}}V_n \text{ for some } n\in [n_{1c}, n_{2c}] \) \notag
\\
& \le P\( V_n^2 \le \(\frac{c}{A}\) n_{2c}^2 \text{ for some } n\in [n_{1c}, n_{2c}] \) \notag
\\
& \le P\( \lvert V_n^2 - \xi^2\rvert \ge \xi^2\epsilon(2-\epsilon) \text{ for some } n\in [n_{1c}, n_{2c}] \)  \notag
\\
& \le P\( \max_{ n_{1c} \le n \le n_{2c} } \l \abs{V_{1n}} + \abs{V_{2n}} + \abs{V_{3n}} + \abs{V_{4n}} \r \ge \xi^2\epsilon(2-\epsilon) \),
\end{align}
where $V_{1n} =\( \frac{ \widehat{\Delta}_n^2 }{ 4 \overline{X}_n^4 } S_n^2 - \frac{ \Delta^2 }{ 4 \mu^4 }\sigma^2 \)$, $V_{2n} =\( \frac{ \widehat{\Delta}_n }{ \overline{X}_n^3 } \widehat{\tau}_n - \frac{ \Delta}{ \mu^3 }\tau \)$, $V_{3n} =\( \frac{ \widehat{\Delta}_n^2 }{ \overline{X}_n^2 } - \frac{ \Delta^2 }{ \mu^2 }\)$, and $V_{4n} =\( \frac{ s_{wn}^2 }{ 4 \overline{X}_n^2 } - \frac{ \sigma_1^2 }{ \mu^2 }\)$. Let $k=\xi^2\epsilon(2-\epsilon)$. Then, \eqref{A31} can be written as $P(N_c \le n_{2c} ) \le P_1 + P_2 + P_3 + P_4$, where
\begin{align*}
P_i = P\(  \max_{ n_{1c} \le n \le n_{2c} } \lvert V_{in} \rvert \ge \frac{k}{4}  \), \quad \text{for }\, i=1, 2, 3, 4.
\end{align*}
First, let us find an upper bound of $P_1$. Let $T_{1n}=  \( \widehat{\Delta}_n^2 - \Delta^2 \)$, $T_{2n} = \( S_n^2 - \sigma^2 \)$, and $T_{3n} = \( \frac{1}{ 4\overline{X}_n^4 } - \frac{1}{4\mu^4}\)$. Note that 
\begin{align}\label{A32}
V_{1n} = T_{1n}T_{2n}T_{3n} + \Delta^2 T_{2n}T_{3n} + \sigma^2 T_{1n}T_{3n} + \frac{1}{\mu^4}T_{1n}T_{2n}  + \frac{\sigma^2}{\mu^4} T_{1n} + \frac{\Delta^2}{\mu^4} T_{2n} + \Delta^2 \sigma^2 T_{3n}, 
\end{align}
Let us consider the first term in the summation of \eqref{A32} and state the following inequalities.
\begin{align}\label{A33} 
P\(   \max_{ n_{1c} \le n \le n_{2c} } \lvert T_{1n}T_{2n}T_{3n} \rvert \ge \frac{k}{28} \) 
& \le \sum_{i=1}^3 P\( \max_{ n_{1c} \le n \le n_{2c} } \lvert T_{in} \rvert \ge \(\frac{k}{28}\)^{\frac{1}{3}} \) \notag
\\
& \le O( n_{1c}^{-r}) + O( n_{1c}^{-r}) + O( n_{1c}^{-r/2}) = O( n_{1c}^{-r/2}). 
\end{align}
The asymptotic orders in \eqref{A33} are obtained by using lemma \ref{A1}, maximal inequality for reverse martingales (Lee, p. 112, 1990), lemma 2.2 of Sen and and Ghosh (1981), and lemma \ref{A2}. The conditions of  lemma \ref{A3} are also used in \eqref{A33}. Following the same argument as above, one can show that the aymptotic order of probability of large deviations (as in \eqref{A33}) corresponding to the remaining six terms in the summation of \eqref{A32} are either $O( n_{1c}^{-r})$ or $O( n_{1c}^{-r/2})$. Therefore,  
\begin{align}\label{A34}
P_1 = P\(  \max_{ n_{1c} \le n \le n_{2c} } \lvert V_{1n} \rvert \ge \frac{k}{4}  \) \le O( n_{1c}^{-r/2}). 
\end{align}
Note that all the estimators in $V_{2n}$ and $V_{3n}$ are U-statistics as we had in the case of $V_{1n}$. So, following similar arguments as in the proof of \eqref{A34}, one can show that both $P_2$ and $P_3$ are $O( n_{1c}^{-r/2})$ as $c \downarrow 0$.  

To work with $P_4$, we note that the expression of $V_{4n}$ involves $s_{wn}^2$ which is not a U-statistics. Therefore, arguments given in the case of $P_1$-$P_3$ may not work without additional result. Following the proof of lemma 3.1 of Sen and Ghosh (1981) and noting that $E\( X_1^{4r}\) < \infty$ for $r \ge 1$,
\begin{align}\label{A35}
P\( \max_{ n_{1c} \le n \le n_{2c} } \left\vert  \frac{s_{wn}^2}{4} -\sigma_1 ^2 \right\vert \ge K \) \le  O( n_{1c}^{-r}), \,\, \text{ for any positive constant $K$}. 
\end{align}
Noting that $V_{4n} = W_{1n}W_{2n} + \sigma_1^2 W_{2n} + \mu^{-2} W_{1n}$, where $W_{1n} =\( \frac{s_{wn}^2}{4} -\sigma_1 ^2\)$ and $W_{2n} =\( \frac{1}{\overline{X}_n^2} -\frac{1}{\mu^2}\)$, 
 \begin{align}\label{A36}
P_4 & \le P\( \max_{ n_{1c} \le n \le n_{2c} } \lvert W_{1n}W_{2n} \rvert \ge \frac{k}{12} \) + P\( \max_{ n_{1c} \le n \le n_{2c} } \lvert W_{2n} \rvert \ge \frac{k}{12\sigma_1^2} \) + P\( \max_{ n_{1c} \le n \le n_{2c} } \lvert W_{1n} \rvert \ge \frac{k\mu^2}{12} \) \notag
\\
& \le \sum_{i=1}^2 P\( \max_{ n_{1c} \le n \le n_{2c} } \lvert W_{in} \rvert \ge \sqrt{\frac{k}{12}} \) + O( n_{1c}^{-r/2}) + O( n_{1c}^{-r}) \le O( n_{1c}^{-r/2}).
\end{align}
The asymptotic orders in \eqref{A36} are obtained by using lemma \ref{A2} and the inequality in \eqref{A35}. We complete the proof of (i) by adding all the upper bounds for $P_1$-$P_4$ and noting that $n_{1c} = O\( c^{-1/(2+2\gamma)}\)$. The proof for part (ii) of lemma \ref{A3} is very similar to the proof of part (i).  
\end{proof}
%
%
\begin{lemma}\label{A4}
 If nonnegative i.i.d. observations $X_1, \ldots, X_n$ are such that $E(X_1^{4r})$ and $E(X_1^{-4r\alpha})$ exist for some $r\ge 1$ and $\alpha >1$, then
\[ E\(\max_{ n_{1c} \le n \le n_{2c} } (G_n - G_F)^r \) =  O( n_{1c}^{-r/2}) \,\, \text{ as } c\downarrow 0. \]
\end{lemma}
\begin{proof}
Applying $C_r$ inequality, we can write
\begin{align*}
\left( G_{n}-G_{F}\right) ^{r} & = \left\{ \widehat{\Delta }
_{n}\left( \frac{1}{2\overline{X}_{n}}-\frac{1}{2\mu }\right) +\frac{1}{2\mu
}\left( \widehat{\Delta }_{n}-\Delta \right) \right\} ^{r}
 \\
&  \le  \frac{1}{2} \left\{ \widehat{\Delta }_{n}^{r}\left( \frac{1}{
\overline{X}_{n}}-\frac{1}{\mu }\right) ^{r}+\frac{1}{\mu ^{r}}\left(
\widehat{\Delta }_{n}-\Delta \right) ^{r}\right\} 
\end{align*}
By Cauchy-Schwarz inequality and lemma  9.2.4 of Ghosh et al. (1997), we have
\begin{align}\label{A41}
& 2 E\(\max_{ n_{1c} \le n \le n_{2c} } (G_n - G_F)^r \) 
\\
& \le \l  E\( \max_{ n_{1c} \le n \le n_{2c} } \widehat{\Delta}_n^{2r} \)  E\( \max_{ n_{1c} \le n \le n_{2c} } \( \frac{1}{ \overline{X}_n } - \frac{1}{\mu} \)^{2r} \) \r^{\frac{1}{2}} + \frac{1}{\mu^r} \( \frac{r}{r-1} \)^r E\(  \widehat{\Delta}_{n_{1c}} -\Delta \)^r \notag 
\\
& \le  \l  E\( \max_{ n_{1c} \le n \le n_{2c} } \widehat{\Delta}_n^{2r} \) \r^{\frac{1}{2}} \l \frac{1}{\mu^{2r}}E\( \max_{ n_{1c} \le n \le n_{2c} } (\overline{X}_n -\mu)^{4r} \) E\( \max_{ n_{1c} \le n \le n_{2c} } \frac{1}{ \overline{X}_n^{4r} }  \)\r^{\frac{1}{4}} + O( n_{1c}^{-r/2}).  \notag
\end{align}
The last inequality is obtained by Cauchy-Schwarz inequality and lemma  2.2 of Sen and Ghosh (1981). Note that, by  lemma  9.2.4 of Ghosh et al. (1997), lemma  2.2 of Sen and Ghosh (1981), and existence of $E\(X_1^{4r} \)$,
\begin{align}
& E\( \max_{ n_{1c} \le n \le n_{2c} } (\overline{X}_n -\mu)^{4r} \)  \le \( \frac{4r}{4r-1}\)^{4r} E\( \overline{X}_{n_{1c}} - \mu \)^{4r} \le O( n_{1c}^{-2r}), \label{A42}
\\
& E\( \max_{ n_{1c} \le n \le n_{2c} } \widehat{\Delta}_n^{2r}\)  \le \( \frac{2r}{2r-1}\)^{2r} E\( \widehat{\Delta}_{n_{1c}}^{2r} \) < \infty, \,\, \text{ and} \label{A43} 
\\
& E\( \max_{ n_{1c} \le n \le n_{2c} } \frac{1}{ \overline{X}_n^{4r} } \)  \le 1+ \int_{1}^{\infty}P\( \max_{ n_{1c} \le n \le n_{2c} } \frac{1}{ \overline{X}_n^{4r} } \ge t \) dt \le 1+ \frac{E\( \overline{X}_{n_{1c}}^{-4r\alpha} \) }{\alpha -1} \label{A44}
\end{align}   
is finite as $E\(X_1^{-4r\alpha}\) < \infty$. The last inequality is due to the maximal inequality for reverse submartingales (Lee, p. 112, 1990). Using \eqref{A42}-\eqref{A44} in the upper bound for \eqref{A41}, we complete the proof. 
\end{proof}
%
\begin{lemma}\label{A5}
If $E(X_{1}^{8})$ and
$E(X_{1}^{-\alpha})$ exist for $\alpha > 8$, then $E\left[ \underset{n \ge m} \sup V_n^2 \right] < \infty$ for $m \ge 4$.
\end{lemma}
%
\begin{proof}
To prove lemma \ref{A5}, it is enough to show that: (i) ${ E}\left[ \underset{n\geq m}{\sup } \,s_{wn}^{2}\overline{X}
_{n}^{-2}\right]$, (ii) ${ E}\left[
\underset{n\geq m}{\sup }\left\vert \frac{\widehat{\Delta }_{n}}{\overline{X}
_{n}^{3}}\widehat{\tau }_{n}\right\vert \right]$, (iii) ${ E}\left[ \underset{n\geq m}{\sup }\frac{\widehat{\Delta }
_{n}^{2}}{\overline{X}_{n}^{2}}\right]$, and  (iv)  ${ E
}\left[ \underset{n\geq m}{\sup }\frac{\widehat{\Delta}_n^2}{\overline{X}_{n}^{4}}S_{n}^{2}\right]$ are finite. Following Sen and Ghosh (p. 338, 1981), we have $E\left[ \underset{n\geq m}{\sup }\, s_{wn}^4\right] < \infty$ if $E[X_1^{\alpha}] < \infty$ for $\alpha >4$ and $m \ge 4$. By \eqref{A44},  $E\left[ \underset{n\geq m}{\sup }\, \overline{X}_n^{\, -4} \right] < \infty$ if $E[X_1^{-\alpha}] < \infty$ for $\alpha >4$. Therefore,
\begin{align*}
E\( \underset{n\geq m}{\sup } \, s_{wn}^2 \overline{X}_n^{\, -2} \) \le \l E\( \underset{n\geq m}{\sup } \, s_{wn}^4\) E\(\underset{n\geq m}{\sup }\, \overline{X}_n^{\, -4} \)  \r^{1/2} < \infty.
\end{align*}
For (ii), we note that $\widehat{\Delta}_n$ and $\widehat{\tau}_n$ are U-statistics. Using lemma 9.2.4 of  Ghosh et al. (1997), 
\[ E\( \underset{n\geq m}{\sup }\left\vert \widehat{\Delta}_n^4 \right\vert \) \le \(\frac{4}{3}\)^4 E\( \left\vert \widehat{\Delta}_m^4 \right\vert \) \,\, \text{and}\,\, 
E\( \underset{n\geq m}{\sup }\left\vert \widehat{\tau}_n^4 \right\vert \) \le \(\frac{4}{3}\)^4 E\( \left\vert \widehat{\tau}_m^4 \right\vert \).
\]
Applying Cauchy-Schwarz inequality twice,
\[
 E\( \underset{n\geq m}{\sup }\left\vert \frac{\widehat{\Delta }_{n}}{\overline{X}
_{n}^{3}}\widehat{\tau }_{n}\right\vert \) 
\le \l E\( \underset{n\geq m}{\sup }\left\vert \widehat{\Delta}_n^4 \right\vert \) E\( \underset{n\geq m}{\sup }\left\vert \widehat{\tau}_n^4 \right\vert \) \r^{\frac{1}{4}} \l E\( \underset{n\geq m}{\sup }\left\vert \overline{X}_n^{\,-6} \right\vert \) \r^{\frac{1}{2}} < \infty,
\]
if $E(X_1^8)$ and $E(X_1^{-\alpha})$ exist for $\alpha > 6$. Similarly, we can show that ${ E}\left[ \underset{n\geq m}{\sup }\frac{\widehat{\Delta }
_{n}^{2}}{\overline{X}_{n}^{2}}\right] < \infty$ if  $E(X_1^4)$ and $E(X_1^{-\alpha})$ exist for $\alpha > 4$. Finally,  ${ E
}\left[ \underset{n\geq m}{\sup }\frac{\widehat{\Delta}_n^2}{\overline{X}_{n}^{4}}S_{n}^{2}\right] < \infty$ if $E(X_1^8)$ and $E(X_1^{-\alpha})$ exist for $\alpha > 8$. This completes the proof. 
\end{proof}
%
\begin{lemma}\label{A6}
Let $U_n$ be a U-statistics for estimating $\theta$ based on $n$ observations. For any $\epsilon \in(0, 1)$,
\[ { E}\left( \underset{n_{2c}\leq n\leq n_{c}}{\max }\left( { U}_{n}{ -U}_{n_{c}}\right) ^{4}\right) { =O}\left( \frac{\epsilon}{n_c^2}\right) \,\, \text{ as } \, c \downarrow 0.\]
\end{lemma}
%
\begin{proof}
Since $\l U_n - U_{n_c}\r_{n=n_{2c}}^{n_c}$ is a reverse martingale, lemma 9.2.4 of Ghosh et al. (1997) yields
\begin{align}\label{A61}
{ E}\left( \underset{n_{2c}\leq n\leq n_{c}}{\max }\left( { U}_{n}{ -U}_{n_{c}}\right) ^{4}\right) \le \(\frac{4}{3}\)^4 E\( U_{n_{2c}} - U_{n_c}\)^4. 
\end{align}  
Let $V_n = U_n - \theta$. Using reverse martingale property of $V_n$, i.e., $ E(V_{n_{2c}}\, |\, \mathcal{F}_{n_c}) = V_{n_c},$ we have 
\begin{align}
& E\( V_{n_{2c}}V_{n_c}^3\) = E\(V_{n_c}^4\), \quad E\( V_{n_{2c}}^3 V_{n_c}\) \ge E\(V_{n_c}^4\), \,\, \text{ and }\label{A62} \\
& E\( V_{n_{2c}}^2 V_{n_c}^2\) \le \l E\(V_{n_{2c}}^4\) E\(V_{n_c}^4\) \r^{\frac{1}{2}} \le  E\( V_{n_{2c}}^4\). \label{A63}
\end{align}
Using \eqref{A62}-\eqref{A63} and asymptotic form of $4^{th}$ central moment of U-statistics (Sen, p. 55, 1981),
\begin{align}\label{A64}
E\( U_{n_{2c}} - U_{n_c}\)^4 &= E\( V_{n_{2c}}^4 \) + E\( V_{n_c}^4 \) - 4E\( V_{n_{2c}}V_{n_c}^3\) - 4E\( V_{n_{2c}}^3 V_{n_c}\) + 6 E\( V_{n_{2c}}^2 V_{n_c}^2\) \notag 
\\
& \le 7 \l   E\( V_{n_{2c}}^4 \) - E\( V_{n_c}^4 \) \r = O\( \frac{1}{n_{2c}^2} - \frac{1}{n_c^2} \) + o\(\frac{1}{n_c^2} \) = O\left( \frac{\epsilon}{n_c^2}\right).
\end{align}
\eqref{A64} is obtained by noting that $n_{2c} = n_c(1-\epsilon)$. Hence, the proof is complete.
\end{proof}
%
\begin{lemma}\label{A7}
 If nonnegative i.i.d. observations $X_1, \ldots, X_n$ are such that $E(X_1^{8})$ and $E(X_1^{-16\alpha})$ exist for $\alpha >1$, then for $\epsilon \in(0, 1)$,
\[ { E}\left[ \underset{n_{2c}\leq n\leq n_{3c}}{\max }\left( G_n - G_{n_c} \right)^2\right] { =O}\left( \frac{\sqrt{\epsilon}}{n_c}\right) \,\, \text{ as } \, c \downarrow 0.\]
\end{lemma}
%
\begin{proof}
$ E\left[ \underset{n_{2c}\leq n\leq n_{3c}}{\max }\left( G_n - G_{n_c} \right)^2\right] \le E_1 + E_2$, where $E_2 = E\left[ \underset{n_{c}\leq n\leq n_{3c}}{\max }\left( G_n - G_{n_c} \right)^2\right]$, and
\begin{align}\label{A71}
E_1 & = E\left[ \underset{n_{2c}\leq n\leq n_{c}}{\max }\left( G_n - G_{n_c} \right)^2\right] \notag
\\
& = E \left[   \underset{n_{2c}\leq n\leq n_{c}}{\max } \l \( \frac{1}{\overline{X}_n} -  \frac{1}{\overline{X}_{n_c}} \) \frac{\widehat{\Delta}_n}{2}  +  \frac{1}{2 \overline{X}_{n_c}}\( \widehat{\Delta}_n - \widehat{\Delta}_{n_c} \) \r^2 \right] \le \frac{E_{11} + E_{12}}{4},
\end{align}
where $E_{11} =  E \left[   \underset{n_{2c}\leq n\leq n_{c}}{\max } \( \frac{1}{\overline{X}_n} -  \frac{1}{\overline{X}_{n_c}} \)^2 \widehat{\Delta}_n^2  \right]$ and $E_{12} = E \left[   \underset{n_{2c}\leq n\leq n_{c}}{\max } \frac{1}{ \overline{X}_{n_c}^2 }\( \widehat{\Delta}_n - \widehat{\Delta}_{n_c} \)^2 \right]$. Applying Cauchy-Schwarz inequality thrice, we can write
\[ 
E_{11} \le \l  E \left[   \underset{n_{2c}\leq n\leq n_{c}}{\max } \( \overline{X}_n -\overline{X}_{n_c} \)^4 \right] \r^{\frac{1}{2}} \l  E \left[   \underset{n_{2c}\leq n\leq n_{c}}{\max } \widehat{\Delta}_n^8 \right] \r^{\frac{1}{4}} \l E\( \frac{1}{ \overline{X}^{16}_{n_c}} \) E\( \underset{n_{2c}\leq n\leq n_{c} }{\max } \frac{1}{\overline{X}^{16}_n }\) \r^{ \frac{1}{8} }. 
\]
Using lemma \ref{A6}, lemma 9.2.4 of Ghosh et al. (1997), \eqref{A44}, and the conditions of lemma \eqref{A7}, we conclude that $E_{11} = O\( \sqrt{\epsilon}/ n_c \)$. Similarly, using Cauchy-Schwarz inequality and lemma \ref{A6}, we have $E_{12} = O\( \sqrt{\epsilon}/ n_c \)$. Therefore, $E_1 = O\( \sqrt{\epsilon}/ n_c \)$. Following the same arguments as above, one can show that $E_2 = O\( \sqrt{\epsilon}/ n_c \)$. Hence, lemma \ref{A7} is proved. 
\end{proof}
%

\subsection{Proof of Theorem \ref{thm:main}}
The proof of the parts (i) and (ii) are similar to \cite{Chde2014}.
\emph{(i)} The definition of stopping rule $N_c$ in \eqref{stopping-rule} yields
\begin{align}\label{T1}
\sqrt{\frac{A}{c}}\, V_{N_{c}} \, \le N_c  \, \le m \, + \, \sqrt{\frac{A}{c}}\left( { V}_{N_{c}-1}+{ (N}_{c}{ -1)}^{-\gamma }\right).
\end{align}
Since $N_c \to \infty$ a.s. as $c \downarrow 0$ and $V_n \to \xi$ a.s. as $n \to \infty$, by theorem 2.1 of Gut (\citeyear{gut2009stopped}), $V_{N_c} \to \xi$ a.s.. Hence, dividing all sides of \eqref{T1} by $n_{c}$ and letting $c\rightarrow 0$,
we prove $N_{c}/n_{c}\to 1$ a.s. as $c\downarrow 0$.

\noindent \emph{(ii)} Since $N_c \ge m$ a.s. and $n_c \ge 1$, dividing \eqref{T1} by $n_c$ yields
\begin{align}\label{T2}
N_c /n_c \le m + \frac{1}{\xi} \( \underset{c >0}\sup \, V_{N_c -1} + (m-1)^{-\gamma}\) \,\,\text{ almost surely,}
\end{align}
where $E\( \underset{c >0}\sup V_{N_c -1} \) < \infty$ by lemma \ref{A5}. 
Since  $N_{c}/n_{c}\to 1$ a.s. as $c\downarrow 0$, by the dominated convergence theorem, we conclude that  $\lim_{c \downarrow 0} E(N_c / n_c) =1$. 

\noindent \emph{(iii)} We need to show $\underset{ c\downarrow 0 }\lim R_{N_c}(G_F)/  R^{\ast}_{n_c}(G_F)  = \underset{ c\downarrow 0 }\lim  (A/ 2c n_c) E\( G_{N_c} - G_F\)^2 + \frac{1}{2}\, \underset{c \downarrow 0} \lim E\( N_c / n_c \) =1$. Thus, it is enough to show that $ \underset{ c \downarrow 0} \lim (A / c n_c) E\( G_{N_c} - G_F \)^2 =1$, i.e., $\underset{ c \downarrow 0} \lim \, n_c \,E\( G_{N_c} - G_F \)^2 = \xi^2$. Since we know that $ n_c E\( G_{n_c} - G_F \)^2 = \xi^2$, it is sufficient to show that
\begin{align}\label{T3}
\lim_{c \downarrow 0} n_c  \l  E\( \( G_{N_c} - G_F \)^2 - \( G_{n_c} - G_F \)^2\) \r=0.  
\end{align}
Let $E_1 = E\left[ (G_{N_c} - G_F)^2 I(N_c \le n_{2c}) \right]$. By \eqref{defs}, lemma \ref{A3}, and lemma \ref{A4}, we have
\begin{align}\label{T4}
n_c E_1 & \le E\left[ \underset{ n_{1c} \le n \le n_{2c} } \max \, (G_n - G_F)^2 I(N_c \le n_{2c}) \right] \notag
\\
& \le n_c \l E\left[ \underset{ n_{1c} \le n \le n_{2c} } \max \, (G_n - G_F)^4 \right] P(N_c \le n_{2c})  \r^{\frac{1}{2}} = O\( c^h\), 
\end{align}
where $h = (1- 2\gamma)/(4 + 4\gamma) >0$ using $\gamma \in (0, \frac{1}{2})$. Here, we assume that $E(X_1^{16})$ and $E(X_1^{-16 \alpha})$ exist for $\alpha >1$. Following the same arguments as in lemma \ref{A4}, we can show that $E\( G_{n_c} - G_F \)^4 = O\( n_c^{-2}\)$ provided $E(X_1^{16})$ and $E(X_1^{-16 \alpha})$ exist for $\alpha >1$. Let $E_2 = E\left[ (G_{n_c} - G_F)^2 I(N_c \le n_{2c}) \right]$. By Cauchy-Schwarz inequality and lemma \ref{A3}, we have
\begin{align}\label{T5}
n_c E_2 \le n_c \l E\left[ (G_{n_c} - G_F)^4 \right] P(N_c \le n_{2c})  \r^{\frac{1}{2}} = O\( c^{\frac{1}{1+\gamma}}\)
\end{align}
provided  $E(X_1^{16})$ and $E(X_1^{-24})$ exist. Therefore, combining \eqref{T4} and \eqref{T5}, we have
\begin{align}\label{T6}
\lim_{c \downarrow 0} \, n_c E\left[ \l (G_{N_c} - G_F)^2  - (G_{n_c} - G_F)^2 \r I(N_c \le n_{2c}) \right] = 0.
\end{align}
Using the same arguments as in lemma \ref{A4}, one can show that $E\left[ \underset{n \ge n_{3c}} \max  \( G_n - G_F \)^4 \right] = O\( n_{3c}^{-2}\)$ provided $E(X_1^{16})$ and $E(X_1^{-16 \alpha})$ exist for $\alpha >1$. Let $E_3 = E\left[ (G_{N_c} - G_F)^2 I(N_c \ge n_{3c}) \right]$.  Cauchy-Schwarz inequality and lemma \ref{A3} yields
\begin{align}\label{T7}
n_c E_3  \le n_c \l E\left[ \underset{  n \ge n_{3c} } \max \, (G_n - G_F)^4 \right] P(N_c \ge n_{3c}) \r^{\frac{1}{2}} =  O\( c^{ \frac{1}{2 +2 \gamma} }\).
\end{align}
Following the same approach as in \eqref{T5}, $n_c E\left[  (G_{n_c} - G_F)^2 I(N_c \ge n_{3c}) \right] \le O\( c^{ \frac{1}{2 +2 \gamma} }\)$. Thus,
\begin{align}\label{T8}
\lim_{c \downarrow 0} \, n_c E\left[ \l (G_{N_c} - G_F)^2  - (G_{n_c} - G_F)^2 \r I(N_c \ge n_{3c}) \right] = 0.
\end{align}
Hence, it remains to prove that 
\begin{align}\label{T9}
\lim_{c \downarrow 0} \, n_c E\left[ \l (G_{N_c} - G_F)^2  - (G_{n_c} - G_F)^2 \r I(n_{2c} \le N_c \le n_{3c}) \right] = 0.
\end{align}
Let $W = \left\{ { (G}_{N_{c}}{ -G}_{F}{ )}^{{ 2}}-{ (G}_{n_{c}}{ -G}_{F}{ )}^{{ 2}}\right\}  I(n_{2c}\le
N_c  \le n_{3c} )$. Note that
\begin{align*}
 W &=\left\{ { (G}_{N_{c}}{ -G}_{F}{ )}+{ (G}
_{n_{c}}{ -G}_{F}{ )}\right\}   (G_{N_{c}} -G_{n_{c}} )  I(n_{2c}\le N_c \le n_{3c}) \notag
\\
& \le 2 \left\{ \underset{n_{2c}\leq n\leq n_{3c}}{\max }\left\vert
{ G}_{n}{ -G}_{F}\right\vert \right\} \left\{ \underset{
n_{2c}\leq n\leq n_{3c}}{\max }\left\vert { G_{n}-G}
_{n_{c}}\right\vert \right\}  I(n_{2c}\le N_c \le n_{3c}). 
\end{align*}
Using Cauchy-Schwarz inequality, lemma \ref{A7}, and following the lines of lemma \ref{A4},
\begin{align}\label{T10}
 n_c E[ W] & \le 2n_c \l E \left( \underset{n_{2c}\le n\le n_{3c}} \max \left(  G_n  -G_F \right)^2\right)
 E \left( \underset{n_{2c}\le n\le n_{3c}} \max \left(  G_n - G_{n_{c}}\right)^2 \right) \r^{ \frac{1}{2} } \notag
\\
& \le 2 n_c \l O\( n_c^{-1}\) O\( \frac{\sqrt{\epsilon}}{n_c}\) \r^{\frac{1}{2}} = O\(\epsilon^{1/4}\).
\end{align}
Since \eqref{T10} is true for any $\epsilon \in (0, 1)$, taking limit on both sides of \eqref{T10} as $\epsilon \to 0$, \eqref{T9} is proved. Hence, the proof of theorem \ref{thm:main} is complete.

\subsection{Proof of lemma \ref{lem:asym-risk} }

By bivariate Taylor expansion of $f(\widehat{\Delta}_n, 2\overline{X}_n) = \frac{\widehat{\Delta}_n}{2\overline{X}_n}$ around $(\Delta, 2\mu)$,
\begin{align}\label{L1}
\frac{\widehat{\Delta}_n}{2\overline{X}_n} - \frac{\Delta}{2\mu} = \frac{\widehat{\Delta}_n - \Delta}{2\mu} - \frac{\Delta}{2 \mu^2}(\overline{X}_n - \mu) + R_{1n},
\end{align}
where $R_{1n} = -2(\widehat{\Delta}_n - \Delta)(\overline{X}_n - \mu)/ b^2 \, +\, 4 a (\overline{X}_n - \mu)/ b^3$, $a=\Delta + p(\widehat{\Delta}_n - \Delta)$, $b=2\mu + p(2\overline{X}_n - 2\mu)$, and $p \in (0, 1)$. Let $E_{1n} = E(R_{1n}^2)$,  $E_{2n} = \frac{1}{\mu}E\(R_{1n}\(\widehat{\Delta}_n - \Delta\) \)$, and  $E_{3n} = - \frac{\Delta}{\mu^2}E\(R_{1n}\(\overline{X}_n - \mu\) \)$. Squaring both sides of \eqref{L1} and taking expectation,
\begin{align}\label{L2}
E\( \frac{\widehat{\Delta}_n}{2\overline{X}_n} - \frac{\Delta}{2\mu} \)^2 = \frac{1}{4\mu^2} V(\widehat{\Delta}_n) + \frac{\Delta^2 \sigma^2}{4n\mu^4} - \frac{\Delta}{2\mu^3} cov( \widehat{\Delta}_n, \overline{X}_n ) + \sum_{i=1}^3 E_{in}. 
\end{align}
Using variance and covariance formulas for U-statistics (Lee, 1990), it is simple to show that $V\( \widehat{\Delta}_n \) = \frac{4 \sigma_1^2}{n} + O\( n^{-2}\)$ and $cov( \widehat{\Delta}_n, \overline{X}_n ) = \frac{2}{n}(\tau -\mu \Delta)$. Therefore, it remains to show that $\sum_{i=1}^3 E_{in} = O(n^{-3/2})$. First, we work on $E_{1n}$. Note that $R_{1n}^2 = 4 W_{1n} + 16 W_{2n} -16 W_{3n}$, where $W_{1n} =  \frac{ \(\widehat{\Delta}_n - \Delta\)^2 \(\overline{X}_n - \mu\)^2 }{b^4}$, $W_{2n} = \frac{a^2}{b^6} (\overline{X}_n - \mu)^4$, and $W_{3n} =  \frac{a}{b^5} \(\widehat{\Delta}_n - \Delta\) \(\overline{X}_n - \mu\)^3$. By Cauchy-Schwarz inequality and lemma 2.2 of Sen and Ghosh (1981),
\begin{align*} 
& E\left\vert  W_{1n} I(\overline{X}_n > \mu ) \right\vert \le \frac{1}{16 \mu^4}  E\( \(\widehat{\Delta}_n - \Delta\)^2 \(\overline{X}_n - \mu\)^2 \)  = O(n^{-2}), 
\\
&  E\left\vert  W_{1n} I(\overline{X}_n \le \mu ) \right\vert \le \frac{1}{16} \l E\( \(\widehat{\Delta}_n - \Delta\)^4 \(\overline{X}_n - \mu\)^4 \) E\( \frac{1}{\overline{X}_n^8}\)\r^{\frac{1}{2}} = O(n^{-2}), 
\end{align*}
provided $E(X_1^8)$ and $E(X_1^{-8})$ exist. Following the same approach, we have
\begin{align*} 
 E\left\vert  W_{2n} I(\widehat{\Delta}_n > \Delta) I(\overline{X}_n > \mu ) \right\vert & \le  E\( (\overline{X}_n - \mu)^4 \frac{\widehat{\Delta}_n^2}{(2\mu)^6}   \)  = O(n^{-2}),
\\
 E\left\vert  W_{2n} I(\widehat{\Delta}_n > \Delta) I(\overline{X}_n \le \mu ) \right\vert & \le  E\( (\overline{X}_n - \mu)^4 \frac{\widehat{\Delta}_n^2}{(2\overline{X}_n)^6}   \)  = O(n^{-2}),
\\
 E\left\vert  W_{2n} I(\widehat{\Delta}_n \le \Delta) I(\overline{X}_n > \mu ) \right\vert & \le  \frac{\Delta^2}{(2\mu)^6}  E\( \overline{X}_n - \mu \)^4  = O(n^{-2}),
\\
 E\left\vert  W_{2n} I(\widehat{\Delta}_n \le \Delta) I(\overline{X}_n \le \mu ) \right\vert & \le   E\( (\overline{X}_n - \mu)^4  \frac{\Delta^2}{(2\overline{X}_n)^6}\)  = O(n^{-2}),
\end{align*}
provided $E(X_1^{12})$ and $E(X_1^{-18})$ exist. Similarly, we can show that $E( W_{3n}) =  O(n^{-2})$ provided $E(X_1^{12})$ and $E(X_1^{-20})$ exist. Therefore, $E_{1n} = E( R_{1n}^2) =  O(n^{-2})$.  By Cauchy-Schwarz inequality  and lemma 2.2 of Sen and Ghosh (1981), we obtain $E_{2n} = O(n^{-3/2})$ and $E_{3n} = O(n^{-3/2})$. Hence, lemma \ref{lem:asym-risk} is proved. 
\section{REFERENCES}
Aguirregabiria, V. and Mira, P. (2007), “Sequential estimation of dynamic discrete games,” Econometrica, 75, 1–53.\\
Allison, P. D. (1978), “Measures of inequality,” American Sociological Review, 865–880.\\
Arcidiacono, P. and Jones, J. B. (2003), “Finite mixture distributions, sequential likelihood and the em algorithm,” Econometrica, 71, 933–946.\\
Beach, C. M. and Davidson, R. (1983), “Distribution-free statistical inference with Lorenz curves and income shares,” The Review of Economic Studies, 50, 723–735.\\
Chattopadhyay, B. and De, Shyamal, K. (2014a), “Estimation Accuracy of an Inequality Index,” Recent Advances in Applied Mathematics, Modelling and Simulation, Accepted.\\
Chattopadhyay, B. and De, Shyamal, K. (2014b), “Estimation of Gini Index within Pre-Specied Error,” Submitted.\\
Cochran, W. G. (1977), Sampling techniques, vol. 98, New York, Wiley and Sons.\\
Dantzig, G. B. (1940), “On the Non-Existence of Tests of “Student’s” Hypothesis Having Power Functions Independent of $\sigma$,” The Annals of Mathematical Statistics, 11, 186–192.\\
Davidson, R. (2009), “Reliable inference for the Gini index,” Journal of econometrics, 150, 30–40.\\
Davidson, R. and Duclos, J.-Y. (2000), “Statistical inference for stochastic dominance and for the measurement of poverty and inequality,” Econometrica, 68, 1435–1464.\\
Doob, J. L. (1953), Stochastic processes, New York Wiley.\\
Gastwirth, J. L. (1972), “The estimation of the Lorenz curve and Gini index,” The Review of Economics and Statistics, 306–316.\\
Ghosh, B. K. and Sen, P. K. (1991), Handbook of sequential analysis, vol. 118, CRC Press.\\
Ghosh, M. and Mukhopadhyay, N. (1979), “Sequential point estimation of the mean when the distribution is unspecified,” Communications in Statistics-Theory and Methods, 8, 637–652.\\
Ghosh, M. and Mukhopadhyay, N. (1981), “Consistency and asymptotic efficiency of two stage and sequential estimation procedures,” Sankhy¯a: The Indian Journal of Statistics, Series A, 220–227.\\
Ghosh, M., Mukhopadhyay, N., and Sen, P. K. (1997), Sequential estimation, Wiley (New York).\\
Greene, W. H. (1998), “Gender economics courses in liberal arts colleges: Further results,” The Journal of Economic Education, 29, 291–300.\\
Gut, A. (2009), Stopped random walks: Limit theorems and applications, Springer.\\
Hoeffding, W. (1948), “A class of statistics with asymptotically normal distribution,” The Annals of Mathematical Statistics, 19, 293–325.\\
Hoeffding, W. (1961), “The strong law of large numbers for U-statistics,” Institute of Statistics mimeo series, 302.\\
Hollander, M. and Wolfe, D. A. (1999), Nonparametric statistical methods, New York John Wiley and Sons.\\
Kanninen, B. J. (1993), “Design of sequential experiments for contingent valuation studies,” Journal of Environmental Economics and Management, 25, S1–S11.\\
Lee, A. J. (1990), U-statistics: Theory and Practice, CRC Press.\\
Lo`eve, M. (1963), Probability theory, Van Nostrand, Princeton, NJ.\\
Loomes, G. and Sugden, R. (1982), “Regret theory: An alternative theory of rational choice under uncertainty,” The Economic Journal, 92, 805–824.\\
Mukhopadhyay, N. and De Silva, B. M. (2009), Sequential methods and their applications, CRC Press.\\
Robbins, H. (1959), Sequential estimation of the mean of a normal population. In Probability and Statistics (Harold Cramer Volume), Almquist and Wiksell: Uppsala, pp 235–245.\\
Sen, P. K. (1981), Sequential nonparametrics: Invariance principles and statistical inference, Wiley New York.\\
Sen, P. K. and Ghosh, M. (1981), “Sequential point estimation of estimable parameters based on U-statistics,” Sankhy¯a: The Indian Journal of Statistics, Series A, 331–344.\\
Sproule, R. (1969), “A sequential fixed-width confidence interval for the mean of a U-statistic,” Ph.D. thesis, Ph. D. dissertation, Univ. of North Carolina.\\
Xu, K. (2007), “U-statistics and their asymptotic results for some inequality and poverty measures,” Econometric Reviews, 26, 567–577.
\bibliography{ReferenceList-new}
\bibliographystyle{asa}
\end{document}